\documentclass[runningheads,a4paper]{llncs}
\usepackage{llncsdoc}

\usepackage{amsfonts}
\usepackage{amsmath}
\usepackage{graphicx}
\usepackage{amssymb}
\usepackage{algorithm}
\usepackage[noend]{algorithmic}
\usepackage{subfigure}

\usepackage{amsthm}
\usepackage{balance}
\usepackage[ansinew]{inputenc}
\usepackage{url}
\urldef{\mailsa}\path|{hhshuai, dnyang}@iis.sinica.edu.tw,|
\urldef{\mailsb}\path|psyu@uic.edu,mschen@cc.ee.ntu.edu.tw | 
  
\authorrunning{H.-H. Shuai, D.-N. Yang, P. S. Yu, and M.-S. Chen}

\begin{document}
\title{Scale-Adaptive Group Optimization \\for Social Activity Planning}
\author{Hong-Han Shuai$^1$%
\and De-Nian Yang$^2$\and Philip S. Yu$^3$\and Ming-Syan Chen$^{1,2}$}
\institute{Graduate Institute of Communication Engineering, National Taiwan University \and Institute of Information Science, Academia Sinica\and Department of Computer Science, University of Illinois at Chicago\\
\mailsa\\
\mailsb\\}
\maketitle
\vspace{-3mm}
\begin{abstract}
Studies have shown that each person is more inclined to enjoy a group
activity when 1) she is interested in the activity, and 2) many friends with
the same interest join it as well. Nevertheless, even with the interest and
social tightness information available in online social networks, nowadays
many social group activities still need to be coordinated manually. In this
paper, therefore, we first formulate a new problem, named Participant
Selection for Group Activity (PSGA), to decide the group size and select
proper participants so that the sum of personal interests and social
tightness of the participants in the group is maximized, while the activity
cost is also carefully examined. To solve the problem, we design a new
randomized algorithm, named Budget-Aware Randomized Group Selection (BARGS),
to optimally allocate the computation budgets for effective selection of the
group size and participants, and we prove that BARGS can acquire the
solution with a guaranteed performance bound. The proposed algorithm was
implemented in Facebook, and experimental results demonstrate that social groups
generated by the proposed algorithm significantly outperform the baseline
solutions.
\end{abstract}

\institute{}

\section{Introduction}

\label{Introduction}

Studies have shown that two important factors are usually involved in a
person's decision to join a social group activity: (1) interest in the
activity topic or content, and (2) social tightness with other attendees 
\cite{DeutschGDM,KaplanGDM}. For example, if a person who appreciates jazz
music has complimentary tickets for a jazz concert in Rose Theatre, she is
inclined to invite her friends or friends of friends who are also jazzists.
However, even the information on the two factors is now available online,
the attendees of most group activities still need to be selected manually,
and the process will be tedious and time-consuming, especially for a large
social activity, given the complicated social link structure and the diverse
interests of potential attendees.

Recent studies have explored community detection, graph clustering and graph
partitioning to identify groups of nodes mostly based on the graph structure 
\cite{BrandesCD08}. The quality of an obtained community is usually measured
according to its internal structure, together with its external connectivity
to the rest of the nodes in the graph \cite{SeshadhriCD12}%
. Those approaches are not designed for activity planning because it does
not consider the interests of individual users along with the cost of
holding an activity with different numbers of participants. An event which
attracts too few or too many attendees will result in unacceptable loss for
the planner. Therefore, it is important to incorporate the preference of
each potential participant, their social connectivity, and the activity cost
during the planning of an activity.

With this objective in mind, a new optimization problem is
formulated, named \textit{Participant Selection for Group Activity (PSGA)}.
The problem is given a cost function related to the group size and a social
graph $G$, where each node represents a potential attendee and is associated
with an interest score that describes the individual level of
interest. Each edge has a social tightness score corresponding to the mutual
familiarity between the two persons. Since each participant is more inclined
to enjoy the activity when 1) she is interested in the activity, and 2) many
friends with the same interest join as well, the \textit{preference} of a
node $v_i$ for the activity can be represented by the sum of its interest
score and social tightness scores of the edges connecting to other
participants, while the \textit{group preference} is sum of the total
interest scores of all participants and the social tightness scores of the
edges connecting to any two participants. Moreover, the \textit{group utility%
} here is represented by the group preference subtracted by the \textit{%
activity cost} (ex. the expense in food and siting), which is usually correlated to the
number of participants.\footnote{%
Different weighted coefficients can be assigned to the group utility and
activity cost according to the corresponding scenario.} The objective of
PSGA is to determine the best group size and select proper participants, so
that the group utility is maximized. In addition, the induced graph of the
set $F$ of selected participants is desired to be a connected component, so
that each attendee is possible to become acquainted with another attendee
according to a social path\footnote{%
For some group activities, it is not necessary to ensure that $F$ leads to a
connected subgraph, and those scenarios can be handled by adding a virtual
node $v$ connecting to every other node in $G$, and choosing $v$ in $F$ for
PSGA always creates a connected subgraph in $G\cup \left\{ v\right\} $, but $%
F$ may not be a connected subgraph in $G$.}.

One possible approach to solving PSGA is to examine every possible combination
on every group size. However, this enumeration approach of group size $k$
requires the evaluation of $C_{k}^{n}$ candidate groups, where $n$ is the
number of nodes in $G$. Therefore, the number of group size
and attendee combinations is $O(2^{n})$, and it thereby is not feasible in practical cases.
Another approach is to incrementally construct the group using a greedy
algorithm that iteratively tries each group size and sequentially chooses an
attendee that leads to the largest increment in group utility at each
iteration. However, greedy algorithms are inclined to be trapped in local
optimal solutions. To avoid being trapped in local
optimal solutions, randomized algorithms have been proposed as a simple but
effective strategy to solve problems with large instances \cite%
{MitzenmacherRand05}.

A simple randomized algorithm is to randomly choose multiple start nodes
initially. Each start node is considered as a partial solution, and a node
neighboring the partial solution is randomly chosen and added to the partial
solution at each iteration later. Nevertheless, this simple strategy
has three disadvantages. Firstly, a start node that has the potential to
generate final solutions with high group utility does not receive sufficient
computational resources for randomization in the following iterations. More
specifically, each start node in the randomized algorithm is expanded to
only one final solution. Thus, a good start node will usually fail to
generate a solution with high group utility since it only has one chance to
randomly generate a final solution. The second disadvantage is that the
expansion of the partial solution does not differentiate the selection of
the neighboring nodes. Each neighboring node is treated equally and chosen
uniformly at random for each iteration. Even this issue can be partially
resolved by assigning the selection probability to each neighboring node according to
its interest score and the social tightness of incident edges, this
assignment will lead to the greedy selection of neighbors and thus tends to
be trapped in local optimal solutions as well. The third disadvantage is
that the linear scanning of different group sizes is not computationally tractable
for real scenarios as an online social network contains an enormous number
of nodes.

Keeping the above observations in mind, we propose a randomized algorithm,
called \textit{Budget-Aware Randomized Group Selection (BARGS)}, to
effectively select the start nodes, expand the partial solutions, and
estimate the suitable group size. The
computational budget represents the target number of random solutions. Specifically, \emph{BARGS} first selects a group size limit $%
k_{max}$ in accordance with the cost function\footnote{%
For instance, if the largest capacity of available stadiums for a football
game is $20,000$, $k_{max}$ is set as $20,000$.}. Afterward, $m$ start nodes
are selected, and neighboring nodes are properly added to expand the partial
solution iteratively, until $k_{max}$ nodes are included, while the group size corresponding to the largest group utility is acquired finally. Each start node in \emph{BARGS} is expanded to multiple final
solutions according to the\ assigned budget. To properly invest the
computational budgets, each stage of \emph{BARGS} invests more budgets on
the start nodes and group sizes that are more inclined to generate good
final solutions, according to the sampled results from the previous stages.
Moreover, the node selection probability is adaptively assigned in each stage by
exploiting the cross entropy method. In this paper, we show that our
allocation of computation budgets is the optimal strategy, and prove that the
solution acquired by \emph{BARGS }has a guaranteed performance bound.

The rest of this paper is organized as follows. Section \ref{Prilim}
formulates PSGA and surveys related works. Sections \ref{SIIGQ-CBA} explains 
\emph{BARGS} and derives the performance bound. User study and experimental results are
presented in Section \ref{Exp}, and we conclude this paper in
Section \ref{Conclu}.

\section{Preliminary}

\label{Prilim}

\subsection{Problem Definition}

\label{SGQProb} Given a social network $G=(V,E)$, where each vertex $%
v_{i}\in V$ and each edge $e_{i,j}\in E$ are associated with an interest
score $\eta _{i}$ and a social tightness score $\tau _{i,j}$ respectively,
we study a new optimization problem for finding a set $F$ of vertices which
maximizes the \textit{group utility} $U(F)$, i.e., 
\begin{equation}
U(F)=\sum_{v_{i}\in F}(\eta _{i}+\sum_{v_{j}\in F:e_{i,j}\in E}\pi
_{i,j})-\beta C(|F|),
\end{equation}%
where $F$ with $\left\vert F\right\vert \leq k_{max}$ is a connected
subgraph in $G$ to encourage each attendee to be acquainted with another
attendee with at least one social path in $F$, $C$ is a non-negative
activity cost function based on the number of attendees, and $\beta $ is a
weighted coefficient between the preference and cost. For each node $i$, let $%
\eta _{i}+\sum_{v_{j}\in F:e_{i,j}\in E}\pi _{i,j}$ denote the \textit{%
preference} of node $i$ on the social group activity. PSGA is very
challenging due to the tradeoff between interest, social tightness, and the
cost function, while the constraint assuring that $F$ is connected also
complicates this problem because it is no longer able to arbitrarily choose
any nodes from $G$. Indeed, we show that PSGA is NP-hard..

\begin{theorem}
PSGA is NP-Hard.
\end{theorem}
\begin{proof}
\textbf{\ }We prove that PSGA is NP-hard with the reduction from DkS problem \cite{Feige01}. Given a graph $%
G_{D}=(V_{D},E_{D})$, DkS finds a subgraph with $k$ nodes $F_D$ to maximize the density of
the subgraph. In other words, the purpose of DkS is to maximize the number of edges $E(F_D)$ in the subgraph induced by the selected nodes.

For each instance of DkS, we construct an instance for PSGA by letting $G=G_{D}$ and $k_{max}=\infty $, where $\eta _{i}$ of each node $v_{i}\in V$ is set as $0$, $\tau _{i,j}$ of each edge $e_{i,j}\in E$ is assigned as $1$, and $\beta=1$, $C(i)=0$ for $i=k$ and $C(i)=\infty $ for $i \neq k$. Therefore, PSGA will always select $k$ nodes to avoid creating a negative objective value. We first prove
the sufficient condition. For each instance of DkS with solution node set $%
F_{D}$, we let $F=F_{D}$. If the number of edges $E(F_{D})$ in
the subgraph of DkS is $\delta $, the preference of PSGA $W(F)$ is also $\delta$ because $F=F_{D}$ and the optimal group size must be $k$. We then prove the necessary condition. For
each instance of PSGA with $F$, we select the same nodes
for $F_{D}$, and the number of edges $E(F_{D})$ must be maximized since the node number in the solution of PSGA is $k$. The theorem follows. \vspace{-2mm}
\end{proof}

\subsection{Related Works}

\label{RelatedWorks} A recent line of study has been proposed to find
cohesive subgroups in social networks with different criteria, such as
cliques, $n$-clubs, $k$-core, and $k$-plex. Sar\'{\i}y{\"{u}}ce et al. \cite%
{kcoreVLDB13} proposed an efficient parallel algorithm to find a $k$-core
subgraph, where every vertex is connected to at least $k$ vertices in the
subgraph. Xiang et al. \cite{cliqueICDE13} proposed a branch-and-bound
algorithm to acquire all maximal cliques that cannot be pruned during the
search tree optimization. Moreover, finding the maximum k-plexes was
comprehensively discussed in \cite{McClosky12}. On the other hand, community
detection and graph clustering have been exploited to identify the subgraphs
with the desired structures \cite{BrandesCD08}. The quality of a community
is measured according to the structure inside the community and the
structure between the community and the rest of the nodes in the graph, such
as the density of local edges, deviance from a random null model, and
conductance \cite{SeshadhriCD12}. Nevertheless, the above
models did not examine the interest score of each user and the social
tightness scores between users, which have been regarded as crucial factors
for social group activities. Moreover, the activity cost for the group is
not incorporated during the evaluation.

In addition to dense subgraphs, social groups with different characteristics
have been explored for varied practical applications. Expert team formation
in social networks has attracted extensive research interest. The problem of
constructing an expert team is to find a set of people possessing the
required skills, while the communication cost among the chosen friends is
minimized to optimize the rapport among the team members to ensure efficient
operation. Communication costs can be represented by the graph diameter,
the size of the minimum spanning tree, and the total length of the shortest
paths \cite{KargarExperts11}. By contrast, minimizing the total
spatial distance with R-Tree from the group with a given number of nodes to
the rally point is also studied \cite{Yang12}. Nevertheless, this paper
focuses on a different scenario that aims at identifying a group with the
most suitable size according to the activity cost, while those selected
participants also share the common interest and high social tightness.

\section{Algorithm Design for PSGA\label{SIIGQ-CBA}}

To solve PSGA, a baseline approach is to incrementally
constructing the solution by sequentially choosing and adding a neighbor
node that leads to the largest increment in the group preference until $%
k_{max}$ people are selected. Afterward, we derive the group utility for
each $k$ by incorporating the activity cost, $1\leq k\leq k_{max}$, and
extract the group size $k^{\ast}$ with the maximum group utility.

The greedy algorithm, despite the simplicity, the search space of the greedy
algorithm is limited and thus tends to be trapped in a local optimal
solution, because only a single sequence of solutions is explored. To address the
above issues, this paper proposes a randomized algorithm \emph{BARGS} to
randomly choose $m$ start nodes\footnote{%
The impact of $m$ will be studied in Section \ref{Exp}.}. \emph{BARGS }%
leverages the notion of Optimal Computing Budget Allocation (OCBA) \cite%
{OCBA10} to systematically generate the solutions from each start node,
where the start nodes with more potential to generate the final solutions
with large group utility will be allocated with more budgets (i.e., expanded
to more final solutions). In addition, since each start nodes can generate
the final solutions with different group sizes, the size with larger group utility
will be associated with more budgets as well (i.e., generated more times).
Specifically, \emph{BARGS} includes the following two phases.

\emph{1) Selection and Evaluation of Start Nodes and Group Sizes:} This
phase first selects $m$ start nodes according to the summation of the
interest scores and social tightness scores of incident edges. Each start
node acts as a seed to be expanded to a few final solutions. At each
iteration, a partial solution, which consists of only a start node at the
first iteration or a connected set of nodes at any iteration afterward, is
expanded by randomly selecting a node neighboring to the partial solution,
until $k_{\max }$ nodes are included. The group utility of each final
solution is evaluated to optimally allocate different computational budgets
to different start nodes and different group sizes in the next phases.

\emph{2) Allocation of Computational Budgets:} This phase is divided into $r$
stages\footnote{The detailed settings of the parameters of the
algorithm, such as $m$, $r$, $\alpha$, and $\beta$ are presented in
the next section}, while each stage shares the same total computational budget. In
the first stage, the computational budget allocated to each start node is
determined by the sampled group utility in the first phase. In each stage
afterward, the computational budget allocated to each start node is adjusted
by the sampled results in the previous stages. Note that each node can
generate different numbers of final solutions with different group sizes.
The sizes with small group utility sampled in the previous stages will be
associated with smaller computational budgets in the current stage.
Therefore, if the activity cost is a convex cost function, the cost
increases more significantly as the group size grows, and \emph{BARGS} tends
to allocate smaller computational budgets and thus generate fewer final
solutions with large group sizes.

During the expansion of the partial solutions, we differentiate the
probability to select each node neighboring to a partial solution. One
intuitive way is to associate each neighboring node with a different
probability according to the sum of the interest score and social tightness
score on the incident edge. Nevertheless, this assignment is similar to the
greedy algorithm as it limits the scope to only the local information
associated with each node, making it difficult to generate a final solution
with large group utility. By contrast, \emph{BARGS} exploits the cross
entropy method \cite{RubinsteinCE01} according to sampled results in the
previous stages in order to optimally assign a probability to the edge
incident to a neighboring node.

The detailed pseudocode is presented in Algorithm \ref{BARGS-alg}. In the following, we first present how to optimally
allocate the computational budgets to different start nodes and different
group sizes. Afterward, we exploit the cross entropy method to differentiate
the neighbor selection during the expansion of the partial solutions.
Finally, we derive the approximation ratio of the proposed algorithm.

\subsubsection{Allocation of Computational Budgets}

Similar to the baseline greedy algorithm, allocating more computational
budgets to a start node $v_{i}$ with larger group utility (i.e., $\eta
_{i}+\sum_{v_{j}\in F:e_{i,j}\in E}\pi _{i,j}$) examines only the local
information and thus is difficult to generate the solution with large group
utility. Therefore, to optimally allocate the computational budgets for each
start node and size, we first define the solution quality as follows.

\begin{definition}
The solution quality, denoted by $Q$, is defined as the maximum group
utility of the solution generated from the $m$ start nodes among all sizes.
\end{definition}

For each stage $t$ of phase 2 in \emph{BARGS}, let $N_{i,k,t}$ denote the
computational budgets allocated to the start node $v_{i}$ with size $k$ in
the $t$-th stage. In the following, we first derive the optimal ratio of the
computational budgets allocated to any two start nodes $v_{i}$ and $v_{j}$
with size $k$ and $l$, respectively. Let two random variables $Q_{i,k}$ and $%
Q_{i,k}^{\ast }$ denote the sampled group utility of any solution and the
maximal sampled group utility of a solution for start node $v_{i}$ with size 
$k$, respectively. If the activity cost is not considered, according to the
central limit theorem, $Q_{i,k}$ follows the normal distribution when $%
N_{i.k}$ is large, and it can be approximated by the uniform distribution in 
$[c_{i,k},d_{i,k}]$ as analyzed in OCBA \cite{OCBA10}, where $c_{i,k}$ and $%
d_{i,k}$ denote the minimum and maximum sampled group utility in the
previous stages, respectively. On the other hand, when the activity cost is
considered, the cumulative distribution function is shifted by $C(k)$, and
it still follows the same distribution. Therefore, we have the following
lemma.

\begin{lemma}
The probability that the solution generated from the start node $v_{i}$ with
size $k$ is better than the solution generated from the start node $v_{j}$
with size $l$, i.e., $P(Q_{i,k}^{\ast }\leq Q_{j,l}^{\ast })$, is as
follows. 
\begin{equation}
P(Q_{i,k}^{\ast }\leq Q_{j,l}^{\ast })\leq \left\{ \begin{aligned}
0&~~\text{ if }d_{j,l} \leq c_{i,k}.
\\\frac{1}{2}(\frac{d_{j,l}-c_{i,k}}{d_{i,k}-c_{i,k}})^{N_{i,k}}&~~\text{ if
}d_{j,l} \geq c_{i,k} \\ 1&~~\text{ if }d_{i,k} \leq c_{j,l}\end{aligned}%
\right.  \label{usedprob}
\end{equation}%
\label{lemmausedprob}
\end{lemma}

\vspace{-5mm}

\begin{proof}
The cumulative distribution
function of $Q_{i,k}$ is  
\begin{equation*}
P_{Q_{i,k}}(x)=\left\{ \begin{aligned} 0&~~\text{ if }x\leq c_{i,k}. \\
\frac{x-c_{i,k}}{d_{i,k}-c_{i,k}}&~~\text{ if }c_{i,k}\leq x\leq d_{i,k}. \\
1&~~\text{ otherwise.}\end{aligned}\right. 
\end{equation*}
After incorporating the operation cost function $C(|F|)$ with $|F|=k$, the
cumulative distribution function of $Q_{i,k}$ is 
\begin{equation}
P_{Q_{i,k}}(x)=\left\{ \begin{aligned} 0&~~\text{ if }x\leq c_{i,k}-\beta
C(k). \\ \frac{x-c_{i,k}}{d_{i,k}-c_{i,k}}&~~\text{ if }c_{i,k}-\beta
C(k)\leq x\leq d_{i,k}-\beta
C(k). \\ 1&~~\text{ otherwise.}\end{aligned}\right. 
\label{gaussianshift}
\end{equation}

Therefore, for the maximal value $Q_{i,k}^{\ast }$, \vspace{-2pt} 
\begin{equation*}
p_{Q_{i,k}^{\ast}}(x)=N_{i,k}P_{Q_{i,k}}(x)^{N_{i,k}-1}p_{Q_{i,k}}(x),
\end{equation*}%
\vspace{-3mm}
\begin{equation*}
P_{Q_{i,k}^{\ast}}(x)=P_{Q_{i,k}}(x)^{N_{i,k}}.
\end{equation*}%
From Eq. \ref{gaussianshift}, the cumulative distribution function is
shifted by $C(k)$ when we incorporate the operation cost, and it thus still follows the same distribution. Assume that $d_{j,l>}c_{i,k}$, the probability that the solution
generated from the start node $v_{i}$ with size $k$ is better than the
solution generated from the start node $v_{j}$ with size $l$,
i.e., $P(Q_{i,k}^{\ast }\leq Q_{j,l}^{\ast })$, can be derived according to \cite{WASO2013} as follows.
\begin{equation*}
P(Q_{i,k}^{\ast }\leq Q_{j,l}^{\ast })\leq\left\{ \begin{aligned} 0&~~\text{ if }d_{j,l} \leq c_{i,k}. \\\frac{1}{2}(\frac{d_{j,l}-c_{i,k}%
}{d_{i,k}-c_{i,k}})^{N_{i,k}}&~~\text{ if }d_{j,l} \geq c_{i,k} \\ 
1&~~\text{ if }d_{i,k} \leq c_{j,l}\end{aligned}\right. 
\end{equation*}
The Lemma follows.
\end{proof}

\vspace{-1mm}Let $v_{b}$ and $k_{b}^{\ast }$ denote the best start node and
best activity size for $v_{b}$, respectively. With Lemma \ref{lemmausedprob}%
, \emph{BARGS} in each stage allocates the computational budgets to
different start nodes as follows. \vspace{-2mm} 
\begin{equation}
\frac{N_{i,t}}{N_{j,t}}=\frac{P(Q=Q_{i}^{\ast })}{P(Q=Q_{j}^{\ast })},
\label{Eq:Thm1Eq}
\end{equation}%
where $P(Q=Q_{i}^{\ast })=\sum_{k}P(Q_{i,k}^{\ast }\geq Q_{b,k_{b}^{\ast
}}^{\ast })$, and the ratio of the computational budget allocation is
optimal in OCBA \cite{OCBA10}, which implies that any other allocation
generates a smaller $Q$. Note that if the allocated computational budgets
for a start node is $0$ in the $t$-th stage, we prune off the start node in
the any stage afterward. After deriving the computational budget $N_{i,t}$
for each start node $v_{i}$, we distribute the budgets to the solutions with
different group sizes. Let $N_{i,k,t}$ denote the number of solutions with
group size $k$ from the start node $v_{i}$. 
\begin{equation}
N_{i,k,t}=N_{i,t}\frac{P(Q_{i,k}^{\ast }\geq Q_{b,k_{b}^{\ast }}^{\ast })}{%
\sum_{k}P(Q_{i,k}^{\ast }\geq Q_{b,k_{b}^{\ast }}^{\ast })}.  \label{targetn}
\end{equation}%
It is worth noting that when we generate
a solution with size $k$, the solutions from size $1$ to size $k-1$ are also
generated as well. Therefore, to avoid generating an excess number the
solutions with small group sizes, it is necessary to relocate the
computation budgets. Let $\hat{N}_{i,k,t}$ denote the reallocated budget of
start node $v_{i}$ with size $k$ in $t$-th stage. \emph{BARGS} reallocates
the computational budgets from size $k-1$ as follows. 
\begin{equation}
\hat{N}_{i,k,t}=\max (0,N_{i,k,t}-\sum_{l>k}{\hat{N}_{i,l,t}}).
\label{reallocate}
\end{equation}%
Specifically, after deriving $N_{i,k,t}$ with Eq. \ref{targetn}, \emph{BARGS}
derives $\hat{N}_{i,k,t}$ from $k=k_{max}$ to $1$. Initially, $\hat{N}%
_{i,k_{\max },t}=N_{i,k_{\max },t}$. Afterward, for $k=k_{max}-1$, if $%
N_{i,k_{max}-1,t}$ is equal to $N_{i,k_{max},t}$, it is not necessary to
generate additional solutions with size $k_{max}-1$ since they have been
created during the generation of the solutions with size $k_{max}$. In this
case, $\hat{N}_{i,k_{max}-1,t}$ is $0$. Otherwise, \emph{BARGS} sets $\hat{N}%
_{i,k_{max}-1,t}=N_{i,k_{max}-1,t}-N_{i,,k_{max},t}$. The above process
repeats until $k=1$. Since the number of solutions with size $k$ is still $%
N_{i,k,t}$, the computational budget allocation is still optimal as shown in 
Eq. \ref{Eq:Thm1Eq}.

\vspace{-10pt}

\subsubsection{Neighboring Node Differentiation}

To effectively differentiate neighbor selection, \emph{BARGS} takes
advantage of the cross entropy method \cite{RubinsteinCE01} to achieve
importance sampling by adaptively assigning a different probability to each
neighboring node from the sampled results in previous stages. Take start
node $v_{i}$ with size $k$ as an example, after collecting $N_{i,k,1}$
samples $X_{i,k,1},X_{i,k,2},...,$ $X_{i,k,q},$ $...,$ $X_{i,k,N_{i,k,1}}$
generated from start node $v_{i}$, \emph{BARGS} calculates the total group
utility $U(X_{i,k,q})$ for each sample and sorts them in the descending
order, $U_{(1)}\geq ...\geq U_{(N_{i,k,1})}$. Let $\gamma _{i,k,1}$ denotes
the group utility of the top-$\rho $ performance sample, i.e. $\gamma
_{i,k,1}=U_{(\left\lceil \rho N_{i,k,1}\right\rceil )}$ . With those sampled
results, we set the selection probability $p_{i,k,t+1,j}$ of every node $%
v_{j}$ in iteration $t+1$ from the partial solution expanded from node $%
v_{i} $ by fitting the distribution of top-$\rho $ performance samples as
follows.

\begin{definition}
A Bernoulli sample vector, denoted as $X_{i,k,q}=\langle
x_{i,k,q,1},...,x_{i,k,q,j},$ $...,x_{i,k,q,n}\rangle$, is defined to be the 
$q$-th sample vector from start node $v_{i}$, where $x_{i,k,q,j}$ is $1$ if
node $v_{j}$ is selected in the $q$-th sample and $0$ otherwise.
\end{definition}

\vspace{-1mm} 
\begin{equation}
p_{i,k,t+1,j}=\frac{\sum_{q=1}^{N_{i,k,t}}I_{\{U(X_{i,k,q})\geq \gamma
_{i,k,t}\}}x_{i,k,q,j}}{\sum_{q=1}^{N_{i,k,t}}I_{\{U(X_{i,k,q})\geq \gamma
_{i,k,t}\}}},  \label{Eq:ce-opt}
\end{equation}%
where $I_{\{U(X_{i,k,q})\geq \gamma _{i,k,t}\}}$ is $1$ if the group utility
of sample $X_{i,k,q}$ exceeds a threshold $\gamma _{i,k,t}$ $\in $ $\mathbb{R%
}$, and $0$ otherwise. Intuitively, the neighbor that tends to generate a
better solution will be assigned a higher selection probability. As shown in 
\cite{RubinsteinCE01}, the above probability assignment scheme has been
proved to be optimal from the perspective of cross entropy. Eq. \ref%
{Eq:ce-opt} minimizes the Kullback-Leibler cross entropy (KL) distance
between node selection probability $\overrightarrow{p}_{i,k,t+1}$ and the
distribution of top-$\rho $ performance samples, such that the performance
of random samples in $(t+1)$-th stage is guaranteed to be closest to the top-%
$\rho $ performance samples in $t$-th stage.

\begin{figure}[t]
\centering
\vspace{80pt} \hspace{-40pt} \includegraphics[scale=0.4]{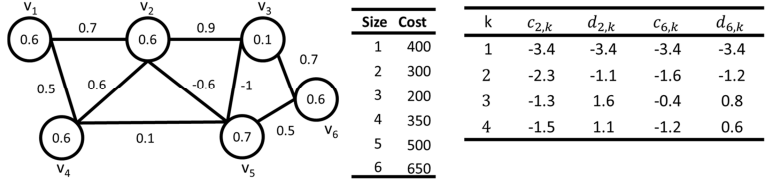} 
\vspace{-88mm} \vspace{-4pt}
\caption{Illustrative example of BARGS}
\label{example}
\end{figure}
\vspace{-10pt}

\begin{example}
Figure \ref{example} presents an illustrative example with a social network
of size $6$. For the greedy algorithm, $v_{5}$ is first selected since its
interest score is the maximum among all nodes, i.e., $0.7$. Afterward, node $%
v_{6}$ is extracted with total preference of $0.7+0.5+0.6=1.8$. $v_{4}$,
instead of $v_{2}$ or $v_{3}$, is chosen because it generates the largest
increment of preference, i.e., $0.7$, and leads to a group with total
preference of $2.5$. After $v_{1}$ is further selected with the increment of 
$1.1$, $v_2$ is selected with total preference of $4.9$. Finally, $v_3$ is
selected with total preference of $5.6$. Assume that the weighting $\beta$
between preference and cost function is $0.01$\footnote{%
The parameter setting of $\alpha$ will be introduced in more details in the
next section.}, the greedy algorithm scans each size to obtain the best
size, i.e., calculating the maximum among $0.7-0.01\cdot 400$, $%
1.8-0.01\cdot 300$, $2.5-0.01\cdot 200$, $3.6-0.01\cdot 350$, $4.9-0.01\cdot
500$, and $5.6-0.01\cdot 650$, and obtains the best size is $3$ with group
utility of $0.5$. In this simple example, the above algorithm is not able to
find the optimal solution since it facilitates the selection of nodes only
suitable at the corresponding iterations.

We also take Figure \ref{example} as an illustrative example for \emph{BARGS}
with $k_{max}=4$. Phase 1 first chooses $\lceil n/k_{max}\rceil =2$ start
nodes by summing up the topic interest score and the social tightness scores
for every node. Therefore, $v_{2}$ with $0.6+0.7+0.6+0.9-0.6=2.2$ and $v_{6}$
with $0.6+0.5+0.7=1.8$ are selected. Next, let $T=20$, $P_{b}=0.7$ and $%
\alpha =0.9$ in this example, and the number of stages is thus $r\leq \frac{%
Tk_{max}\ln \alpha }{n\ln (\frac{2(1-P_{b})}{m-1})}=\frac{20\cdot 4\ln 0.9}{%
6\ln (0.6)}\approx 2$. Each start node generates $5$ samples in the first
stage. The intermediate solution obtained so far is denoted as $V_{S}$, and
the candidate attendees extracted so far is denoted as $V_{A}$. Therefore,
by selecting $v_{2}$ as a start node, the total group utility of $V_{S}$\ $%
=\{v_{2}\}$\ is $0.6-0.01\cdot 400=-3.4$, and $V_{A}=%
\{v_{1},v_{3},v_{4},v_{5}\}$. Since the node selection probability is
homogeneous in the first stage, we randomly select $v_{1}$\ from $V_{A}$\ to
expand $V_{S}$. Now the total group utility of $V_{S}$\ $=\{v_{1},v_{2}\}$
is $U(V_{S})=0.6+0.7+0.6-0.01\cdot 300=-1.1$, and $V_{A}=\{v_{3},v_{4},v_{5}%
\}$. The process of expanding $V_{S}$\ continues until the cardinality of $%
V_{S}$\ reaches $k_{max}=4$, e.g. $v_{5}$ and then $v_{3}$. Afterward, we
record the first sample result $X_{2,2,1}=\langle 1,1,0,0,0,0\rangle $ with
the total group utility of $-1.1$, the worst result of $v_{2}$ with size $2$
($c_{2,2}=-1.1$), and the best result of $v_{2}$ with size $2$ ($%
d_{2,2}=-1.1 $). Similarly, $X_{2,3,1}=\langle 1,1,0,0,1,0\rangle $ with the
total group utility of $-1$ and $X_{2,4,1}=\langle 1,1,1,0,1,0\rangle $ with
the total group utility of $0.7$. The second sampled results from start node 
$v_{2}$\ are $\{v_{2},v_{3},v_{4},v_{1}\}$. Therefore, $X_{2,2,2}=\langle
0,1,1,0,0,0\rangle $ with the total group utility of $-1.4$, $%
X_{2,3,2}=\langle 0,1,1,1,0,0\rangle $ with the total group utility of $0.8$%
, $X_{2,4,2}=\langle 1,1,1,1,0,0\rangle $ with the total group utility of $%
1.2$. Afterward, the worst and the best results of $v_{2}$\ are updated to $%
c_{2,2}=-1.4$, $d_{2,2}=-1.1$, $c_{2,3}=-1$, $d_{2,3}=0.8$, $c_{2,4}=0.7$,
and $d_{2,4}=1.2$. After drawing $3$ more samples from node $v_{2}$, we
repeat the above process for start node $v_{6}$ with 5 samples. The results
are summarized on the right of Figure \ref{example}.

To allocate the computational budgets for the second stage, i.e., $r=2$, we
first find the allocation ratio $N_{2,2}:N_{6,2}$=$\frac{1}{2}((\frac{%
-1.1-(-1.3)}{1.6-(-1.3)})^{5}+1+(\frac{1.1-(-1.3)}{1.6-(-1.3)})^{5})$:\newline$\frac{%
1}{2}((\frac{-1.2-(-1.3)}{1.6-(-1.3)})^{5}$
+$(\frac{0.8-(-1.3)}{1.6-(-1.3)})^{5}+(\frac{0.6-(-1.3)}{1.6-(-1.3)})^{5})$ =%
$1.39:0.32$. Therefore, the allocated computational budgets for start nodes $%
v_{2}$\ and $v_{6}$\ are $\frac{10\cdot 1.39}{1.71}\approx 8$\ and $\frac{%
10\cdot 0.32}{1.71}\approx 2$, respectively. $\hat{N}_{2,2,2}$, $\hat{N}%
_{2,3,2}$, and $\hat{N}_{2,4,2}$ approximate $0$, $\frac{8}{1.388}\approx 6$%
, and $\frac{8\cdot 0.388}{1.388}\approx 2$, respectively. \emph{BARGS}
reallocates the computational budgets by $\hat{N}_{2,3,2}=N_{2,3,2}-\hat{N}%
_{2,4,2}=6$. Afterward, we update the node selection probability. Take the
node selection probability for start node $v_{2}$ with size $3$ in the
second stage for node $v_{1}$ as an example, i.e., $p_{2,3,2,1}$. Given $%
\rho =0.6$, i.e., \emph{BARGS} selects top-$3$ performance samples, if $%
v_{1} $ is selected 2 times in top-3 performance samples, $p_{2,3,2,1}$ is
set as $\frac{2}{3}$. The process for $v_{6}$ is similar and thus omitted
here due to the space constraint. After the second stage, the optimal
solution is $\{v_{1},v_{2},v_{4}\}$ with maximum group utility of $1.6$,
which is better than the group utility generated by the greedy algorithm,
i.e., $0.5$.
\end{example}

\subsubsection{Theoretical Results}

The following theorem first analyzes the probability $P(Q=Q^{\ast}_{b,k_b^{%
\ast}})$ that $v_{b}$, as decided according to the samples in the previous
stages, is actually the start node that generates the maximal group utility
with optimal size $k_b^{\ast}$. Let $\alpha$ denote the closeness ratio
between the maximum of the start node with the maximal group utility and the
maximum of other start nodes or with different sizes, i.e., $%
\alpha=(d_{a,k_{a}^{\ast}}-c_{b,k_{b}^{\ast}})/(d_{b,k_{b}^{%
\ast}}-c_{b,k_{b}^{\ast}})$, where $v_{a}$ generates the maximal group
utility among other start nodes. Therefore, in addition to $0$ and $1$, $%
\alpha$ is allowed to be any other value from $0$ to $1$.

\vspace{-2pt}

\begin{theorem}
\label{ControlComputationThm1} For PSGA with parameter $(m,T,k_{max})$,
where $m$ is the number of start nodes, $T$ is the total computational
budgets, and $k_{max}$ is the group size limit, the probability $%
P(Q=Q_{b,k_{b}^{\ast }}^{\ast })$ that $v_{b}$ selected according to the
previous stages is actually the start node that generates the optimal
solution with optimal size $k_{b}^{\ast }$ is at least $1-\frac{1}{2}%
(k_{max}+m-2)\alpha ^{\frac{T}{rmk_{max}}}$.
\end{theorem}
\begin{proof}
According to the Bonferroni inequality, $p\{\cap _{i=1}^{m}(Y_{i}<0)\}\geq
1-\sum_{i=1}^{m}[1-p(Y_{i}<0)]$. In our case, $Y_{i}$ is replaced by $%
Q^{\ast}_{i,k_i^{\ast}}-Q^{\ast}_{b,k_b^{\ast}}$ to acquire a lower bound for the probability that $v_{b}$ enjoys the maximal group utility with optimal size $k_b^{\ast}$. Therefore, by using Equation \ref{usedprob},
\vspace{-1mm}
\begin{eqnarray*}
&&P(Q=Q^{\ast}_{b,k_b^{\ast}}) \\
&=& p\{\cap _{l=1,l\neq k_b^{\ast}}^{k_{max}}(Q_{b,l}^{\ast}-Q_{b,k_b^{\ast}}^{\ast})\leq 0)\}\cdot \\
&& p\{\cap _{i=1,i\neq b}^{m}(Q^{\ast}_{i,k_i^{\ast}}-Q^{\ast}_{b,k_b^{\ast}}\leq 0)\}\\
&\geq &(1-\sum_{l=1,l\neq k_{b^{\ast}}}^{k_{max}}[1-p(Q^{\ast}_{b,k_l^{\ast}}-Q^{\ast}_{b,k_b^{\ast}}\leq 0)])\cdot\\
&&(1-\sum_{i=1,i\neq b}^{m}[1-p(Q^{\ast}_{i,k_i^{\ast}}-Q^{\ast}_{b,k_b^{\ast}}\leq 0)])\\
&&\\
&\geq &(1-\frac{1}{2}\sum_{l=1,l\neq k_b^{\ast}}^{k_{max}}(\frac{d_{b,l}-c_{b,k_b^{\ast}}}{d_{b,k_b^{\ast}}-c_{b,k_b^{\ast}}})^{N_{b,k_b^{\ast}}})\cdot\\
&&(1-\frac{1}{2}\sum_{i=1,i\neq b}^{m}(\frac{d_{i,k_i^{\ast}}-c_{b,k_b^{\ast}}}{d_{b,k_b^{\ast}}-c_{b,k_b^{\ast}}})^{N_{b,k_b^{\ast}}}
)
\end{eqnarray*}
By introducing $\alpha$, $P(Q=Q^{\ast}_{b,k_b^{\ast}})$ is greater than
\begin{eqnarray*}
&& (1-\frac{1}{2}(k_{max}-1)\alpha^{N_{b,k_b^{\ast}}})(1-\frac{1}{2}(m-1)\alpha^{N_{b,k_b^{\ast}}}) \\
&\geq & 1-\frac{1}{2}(k_{max}+m-2)\alpha^{N_{b,k_b^{\ast}}}\\
&\geq & 1-\frac{1}{2}(k_{max}+m-2)\alpha^{\frac{T}{rmk_{max}}}.
\end{eqnarray*}
The theorem follows.
\end{proof}

Given the total budgets $T$ and a general cost function, i.e., without any
assumption, the following theorem derives a lower bound of the solution
obtained by \emph{BARGS}.

\begin{theorem}
The maximum group utility $E[Q]$ from the solution of \emph{BARGS} is at
least $N_{b,k_{b}^{\ast }}(\frac{1}{N_{b,k_{b}^{\ast }}+1}%
)^{1+N_{b,k_{b}^{\ast }}^{-1}}\cdot Q^{\ast }$, where $N_{b,k_{b}^{\ast }}$
after $r$ stages is $\frac{4+mk_{max}(r-1)}{4rmk_{max}}T$, $Q^{\ast }$ is
the optimal solution for a PSGA problem in $r$-stage computational budget
allocation, and $k_{b}^{\ast }$ is the optimal group size of the best node $%
v_{b}$. \label{baselinethm2}
\end{theorem}

\begin{proof}
It is challenging to derive the performance bound without any assumption on
the cost function due to (1) no useful properties such as such as
monotonicity, submodularity, and convexity, so it is impossible to estimate
the performance according to the size, and (2) the cost function can
dominate the performance bound or be neglected according to $\beta$.
However, the cumulative distribution function of $Q_{i,k}^{\ast}$ follows
the Gaussian distribution regardless to $i$ and $k$. Therefore, we analyze
the performance bound by regarding each combination of $Q_{i,k}^{\ast}$ as a
sampling result of different start nodes.

Notice that, given a fixed size $k$, the maximum preference from the
solution of \emph{BARGS} from the best node $v_{b}$ without the cost
function is at least $N_{b}(\frac{1}{N_{b}+1})^{1+N_{b}^{-1}}\cdot Q^{\ast }$%
, where $N_{b}$ after $r$ stages is $\frac{4+m(r-1)}{4rm}T$, and $Q^{\ast }$
is the optimal solution for a PSGA problem without cost function in $r$%
-stage computational budget allocation. Therefore,

\begin{equation}
E[Q] \geq N_{b,k_b^{\ast}}(\frac{1}{N_{b,k_b^{\ast}}+1})^{1+N_{b,k_b^{%
\ast}}^{-1}}\cdot Q^{\ast },
\end{equation}
If the computational budget allocation is $r-$stages with $T\geq $ $k_{max}mr%
\frac{\ln (k_{max}m-1)}{\ln (\frac{1}{\alpha })}$, $N_{b,k_{b}^{\ast}}$ is $%
\frac{T}{rmk_{max}}+\frac{1}{2}\frac{r-1}{2r}T$, which is $\frac{%
4+mk_{max}(r-1)}{4rmk_{max}}T$. The theorem follows.
\end{proof}

\textbf{Time Complexity of BARGS.} The time complexity of \emph{BARGS}
contains two parts. The first phase selects $m$ start nodes with $O(E+n+$ $%
m\log n)$ time, where $O(E)$ is to sum up the interest and social tightness
scores, $O(n+m\log n)$ is to build a heap and extract $m$ nodes with the
largest sum. Afterward, the second phase of \emph{BARGS} includes $r$
stages, and each stage allocates the computational resources with $O(m)$
time and generates $O(\frac{T}{r})$ new partial solutions with at most $%
k_{max}$ nodes for all start nodes. Therefore, the time complexity of the
second phase is $O\left( r(m+\frac{T}{r}k_{max})\right) =O(k_{max}T)$, and 
\emph{BARGS} therefore needs $O(E+m\log n+k_{max}T)$ running time.


\section{Experimental Results}

\label{Exp}

\subsection{Experiment Setup}

\label{DataPreparation} We implement \emph{BARGS} in Facebook and invite 50
people from various communities, e.g., schools, government, technology
companies, and businesses to join our user study. We compare the solution
quality and running time of manual coordination and \emph{BARGS} for answering PSGA problems, to evaluate the need of an automatic group recommendation service.
Each user is asked to plan 5 social activities with the social graphs
extracted from their social networks in Facebook. The interest scores follow
the power-law distribution with the exponent as 2.5 according to the recent analysis \cite{PowerLawClauset09} on real datasets. The social tightness score between two friends is derived according to the number of common friends, which represents the proximity interaction \cite{ChaojiNnodeLink12}, and the
probability of negative weights \cite{JureSN10}. Then, the weighted coefficient $\lambda$ on social tightness scores and interest scores and the weighted coefficient $\beta$ on group preference and activity cost in Footnote 4 are set as the average value specified by the 50 people, i.e., $\lambda=0.527$ and $\beta=0.514$. Most importantly, after the scores
are returned by the above renowned models, each user is allowed to
fine-tune the two scores by themselves. In addition to the user study, three real datasets are evaluated in the
experiment. The first dataset is crawled from Facebook with $90,269$\ users
in the New Orleans network\footnote{\url{
http://socialnetworks.mpi-sws.org/data-wosn2009.html}}. The second dataset
is crawled from DBLP dataset with $511,163$ nodes and $1,871,070$ edges. The
third dataset, Flickr\footnote{\url{http://socialnetworks.mpi-sws.org/data-imc2007.html}}, with $1,846,198$ nodes and $22,613,981$ edges, is also incorporated to demonstrate the
scalability of the proposed algorithms.

In this paper, the activity cost is modelled by a piecewise linear
function, which can approximate any non-decreasing functions. We set the activity cost according to the auditorium cost and other related cost in
Duke Energy Center\footnote{\url{http://www.dukeenergycenterraleigh.com/uploads/venues/rental/5-rateschedule.pdf}}.
\begin{equation*}
C(k)=\left\{ \begin{aligned} 400-k &~~ \text{ if }0\leq k\leq 100.\\ 850-k
&~~ \text{ if }100< k\leq 600.\\ 2200-k &~~ \text{ if }600< k\leq 1750.\\
\end{aligned}\right.
\end{equation*}

We compare deterministic greedy (\emph{DGreedy}), randomized
greedy (\emph{RGreedy}), and \emph{BARGS} in an HP DL580 server with four
Intel E7-4870 2.4 GHz CPUs and 128 GB RAM. \emph{RGreedy} first chooses the same $m$ start nodes as \emph{BARGS}. At each iteration, \emph{RGreedy} calculates the preference increment of adding a neighboring node $v_j$ to the intermediate solution $V_S$ obtained so far for each neighboring node, and sums them up as the total preference increment. Afterward, \emph{RGreedy} sets the node selection probability of each neighbor as the ratio of the corresponding preference increment to the total preference increment, similar to the concept in the greedy algorithm. Notice that the computation budgets represent the number of generated solutions. With more computation budgets, \emph{RGreedy} generates more solutions of group size $k_{max}$, examines the group utility by subtracting the activity cost from group size $1$ to $k_{max}$, and selects the group with maximum group utility. It is worth noting that \emph{RGreedy} is computationally intensive and not scalable to support a large group size because it is necessary to sum up the interest scores and social tightness scores during the selection of a node neighboring to each partial solution. Therefore, we can only present the results of 
\emph{RGreedy} with small group sizes.

The default $m$ in the experiment is set as $n/k_{max}$ since $n/k_{max}$ groups can be acquired from a network with $n$ nodes if each group has $k_{max}$ participants. The default cross-entropy parameters $\rho$ and $\alpha$ are set as 0.3 and 0.99 as recommended by the cross-entropy method 
\cite{RubinsteinCE01}. Since \emph{BARGS} natively supports parallelization, we
also implemented them with OpenMP for parallelization, to demonstrate the
gain in parallelization with more CPU cores.

\vspace{-9pt}
\subsection{User Study}
\vspace{-4pt}
\label{UserStudy}
Figures \ref{exp1}(a)-(c) compare manual coordination and \emph{BARGS} in the user study. In addition, the optimal solution is also derived with the enumeration method since the network size is very small. Figures \ref{exp1}(a) and (b) present the solution
quality and execution time with different network sizes. The result
indicates that the solutions obtained by \emph{BARGS} are identical to the
optimal solutions, but users are not able to acquire the optimal solutions even when $n=5$. As $n$ increases, the solution quality of manual coordination degrades rapidly. We also compare the accuracy of selecting the optimal group size in Figure \ref{exp1}(c). As $n$ increases, it becomes more difficult for a user to correctly identify the optimal size, while \emph{BARGS} can always select the optimal one. Therefore, it is desirable to deploy \emph{BARGS} as an automatic group recommendation service, especially to address the need of a large group in a massive social network nowadays.

\begin{figure}[t]
\vspace{-3mm}
\centering
\subfigure[] {\
\includegraphics[scale=0.14]{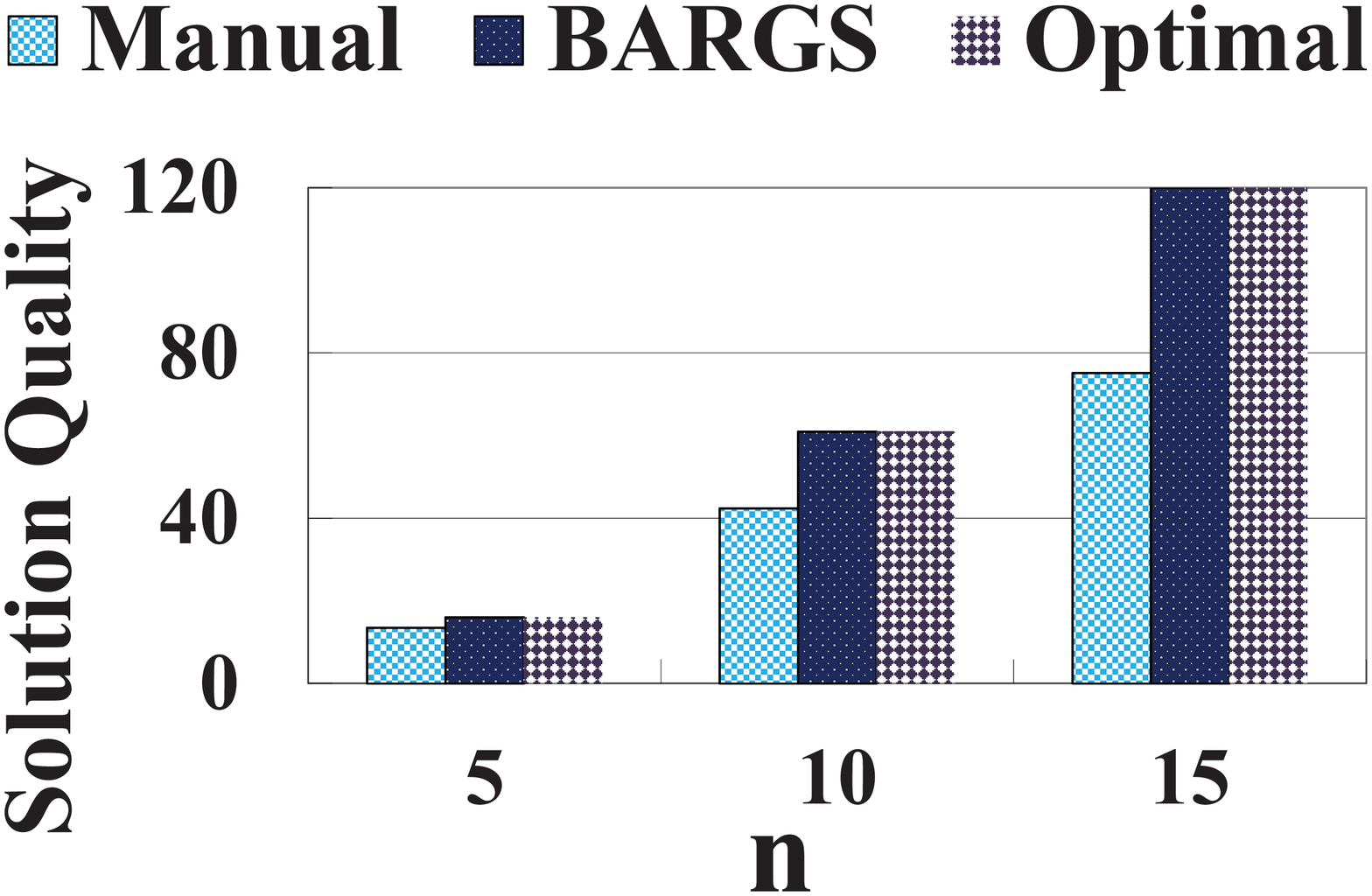} } 
\subfigure[] {\
\includegraphics[scale=0.14]{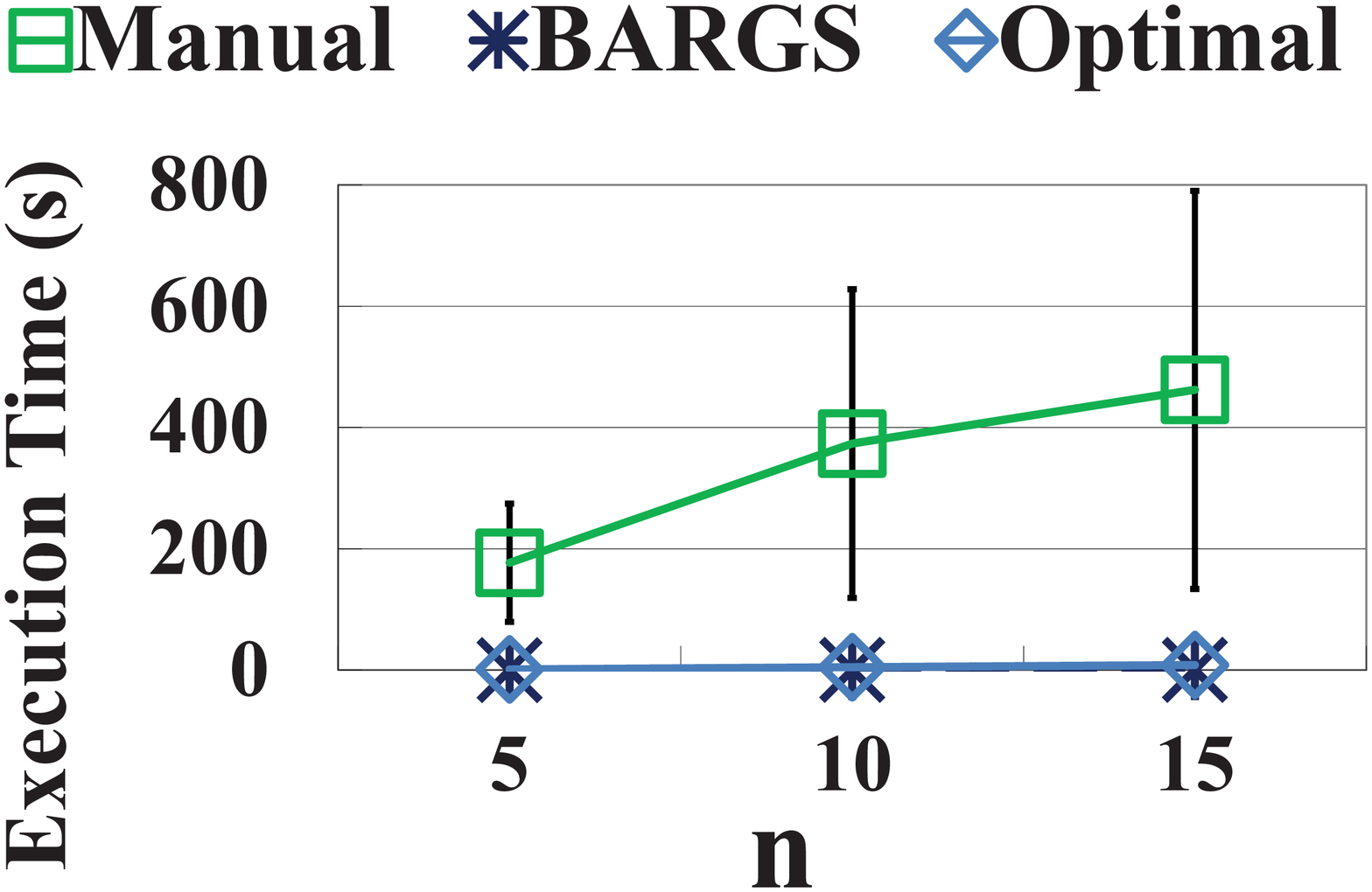} } 
\subfigure[] {\
\includegraphics[scale=0.14]{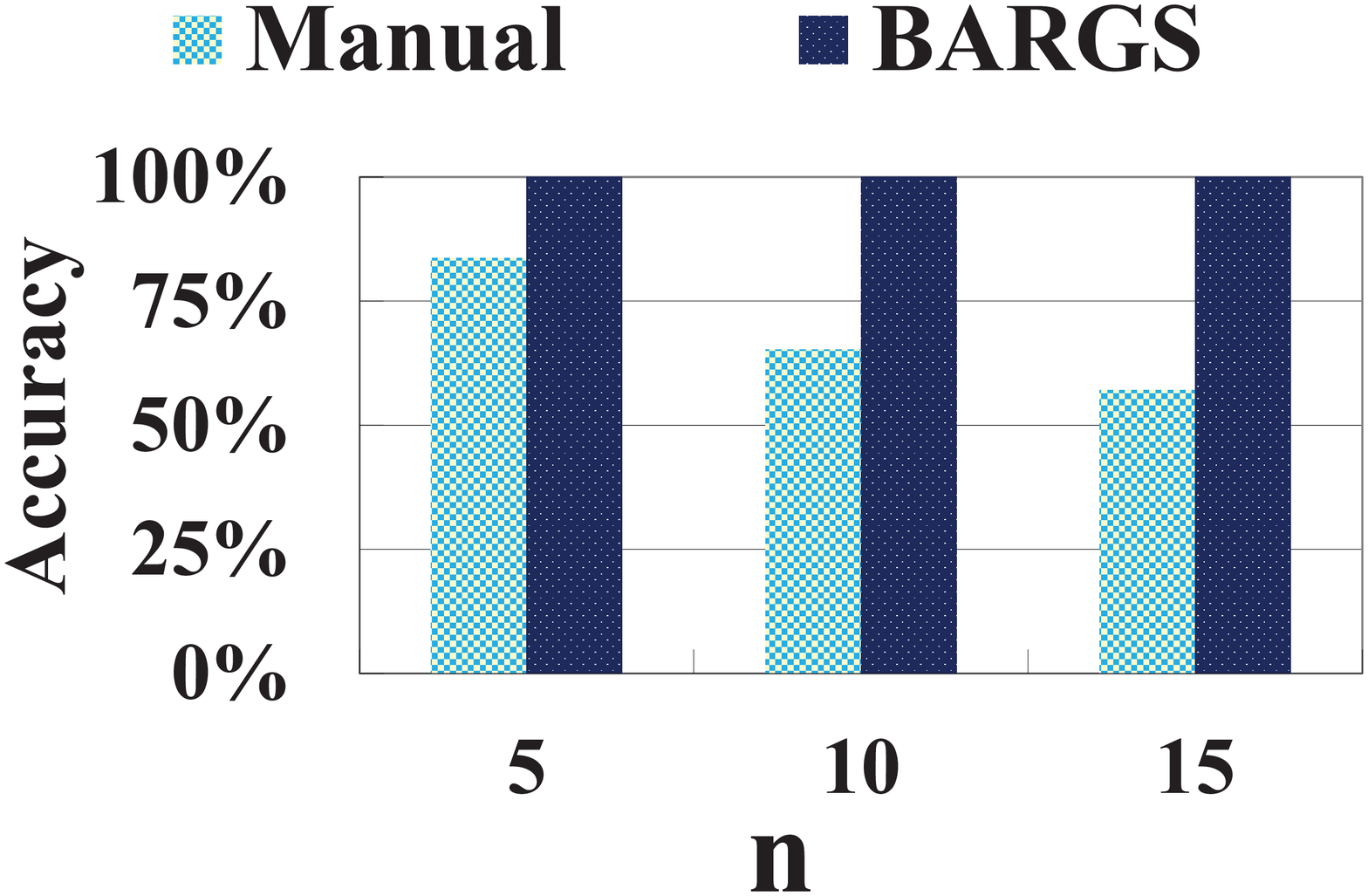} } 
\vspace{-12pt} 
\caption{Results of user study}
\vspace{-7mm}
\label{exp1}
\end{figure}

\subsection{Performance Comparison and Sensitivity Analysis}
\vspace{-2mm}
\begin{figure}[t]
\vspace{-3mm}
\centering
\subfigure[] {\
\includegraphics[scale=0.14]{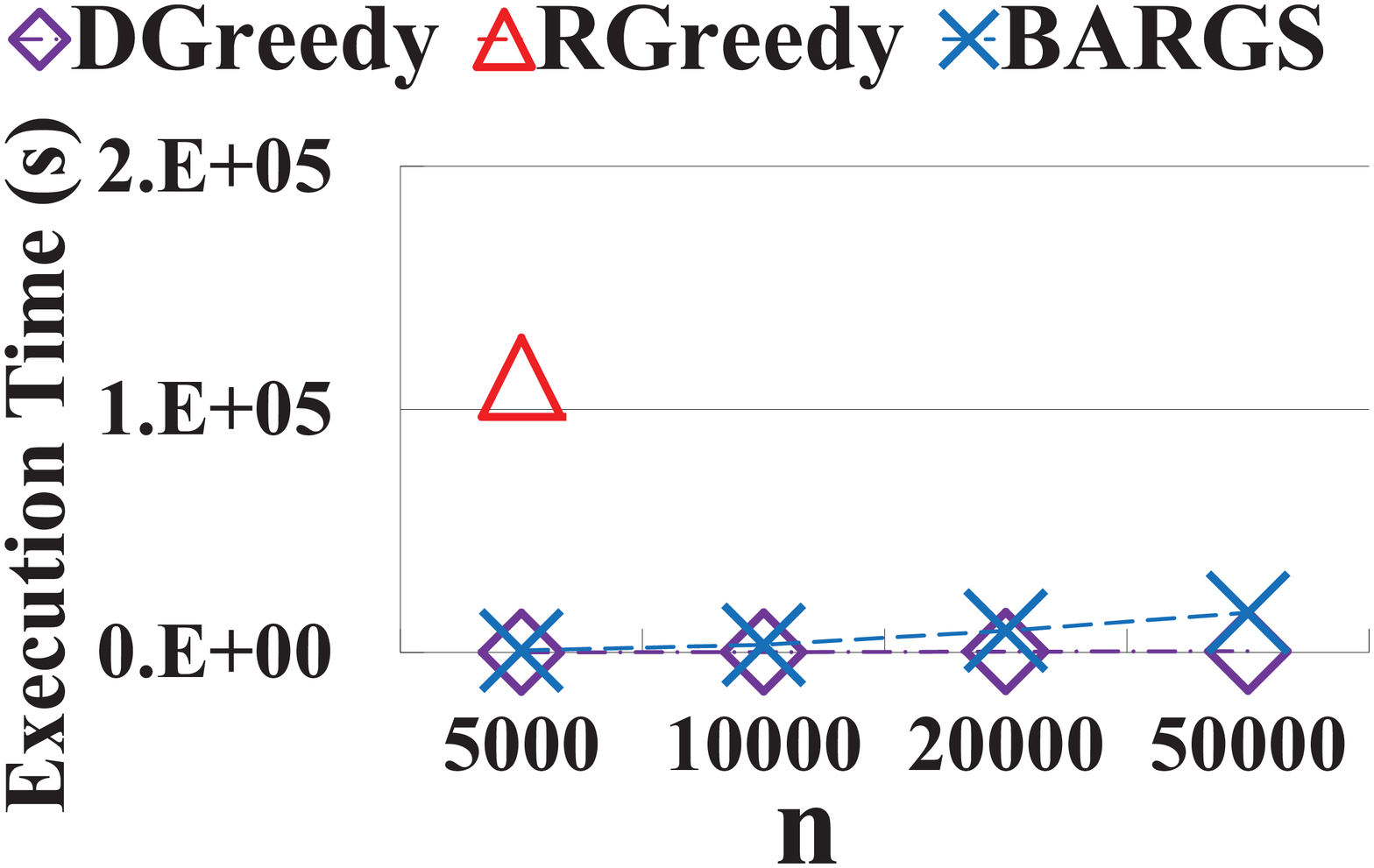} } 
\subfigure[] {\
\includegraphics[scale=0.14]{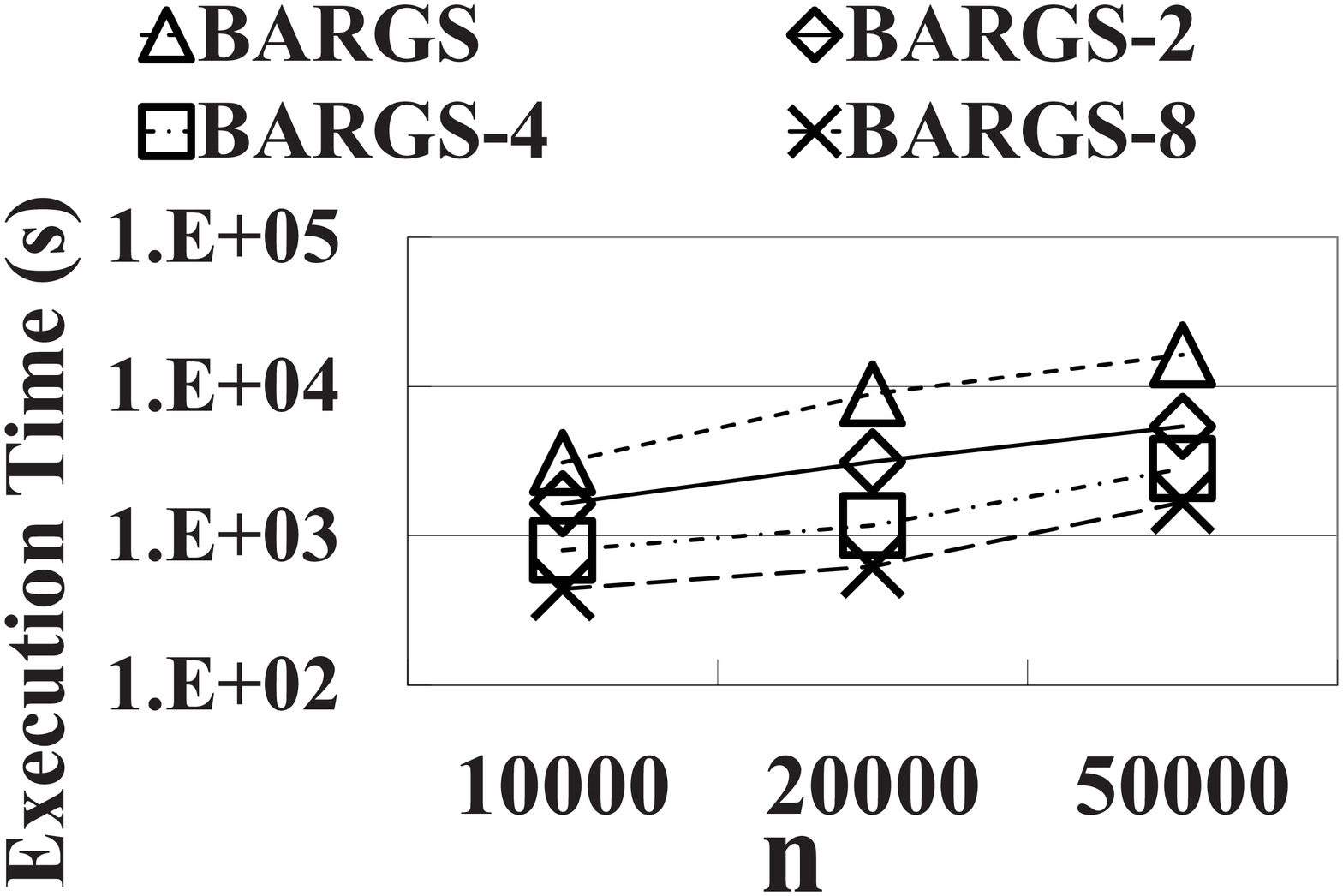} } 
\subfigure[] {\
\includegraphics[scale=0.14]{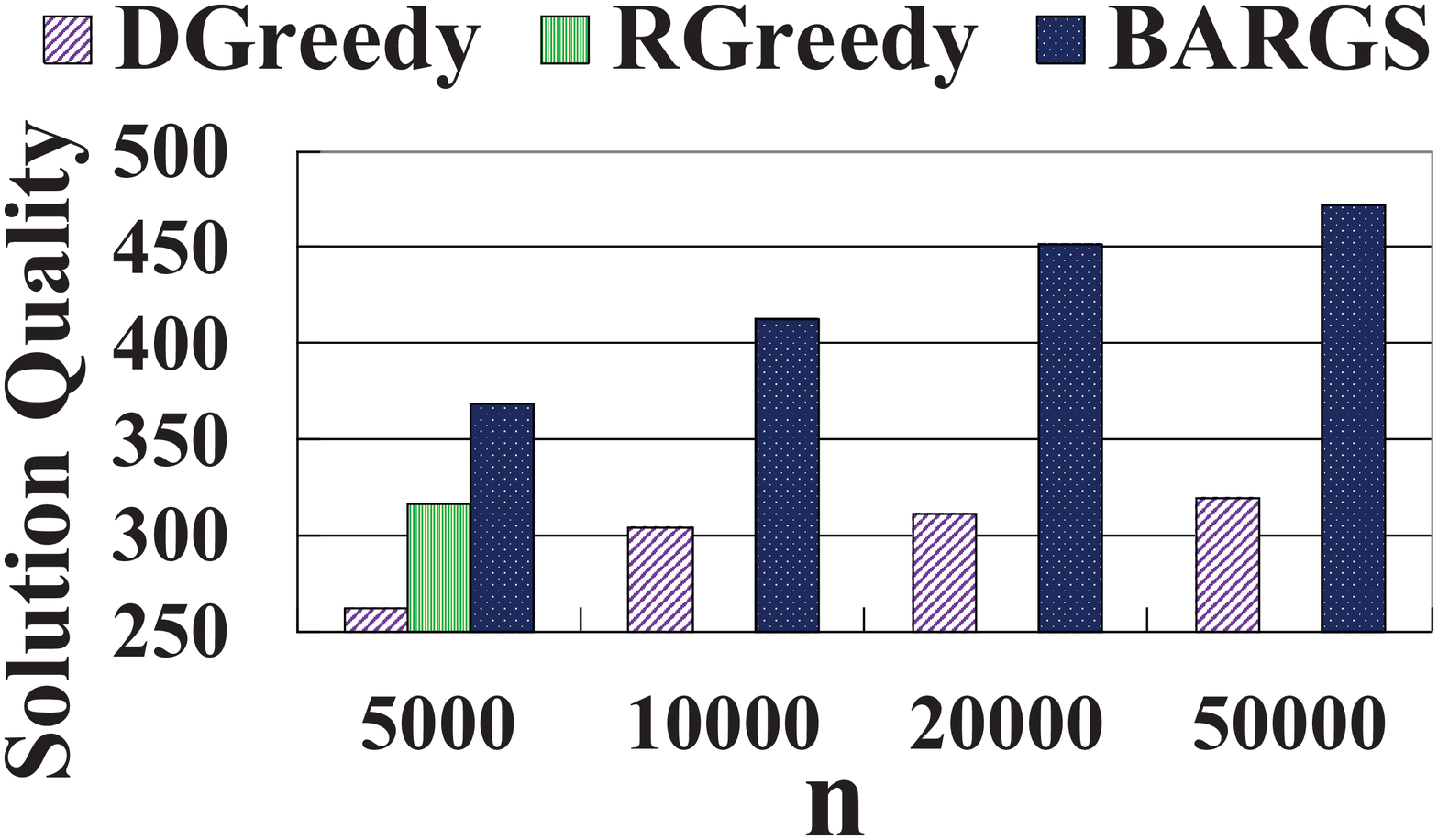} } 
\subfigure[] {\
\includegraphics[scale=0.14]{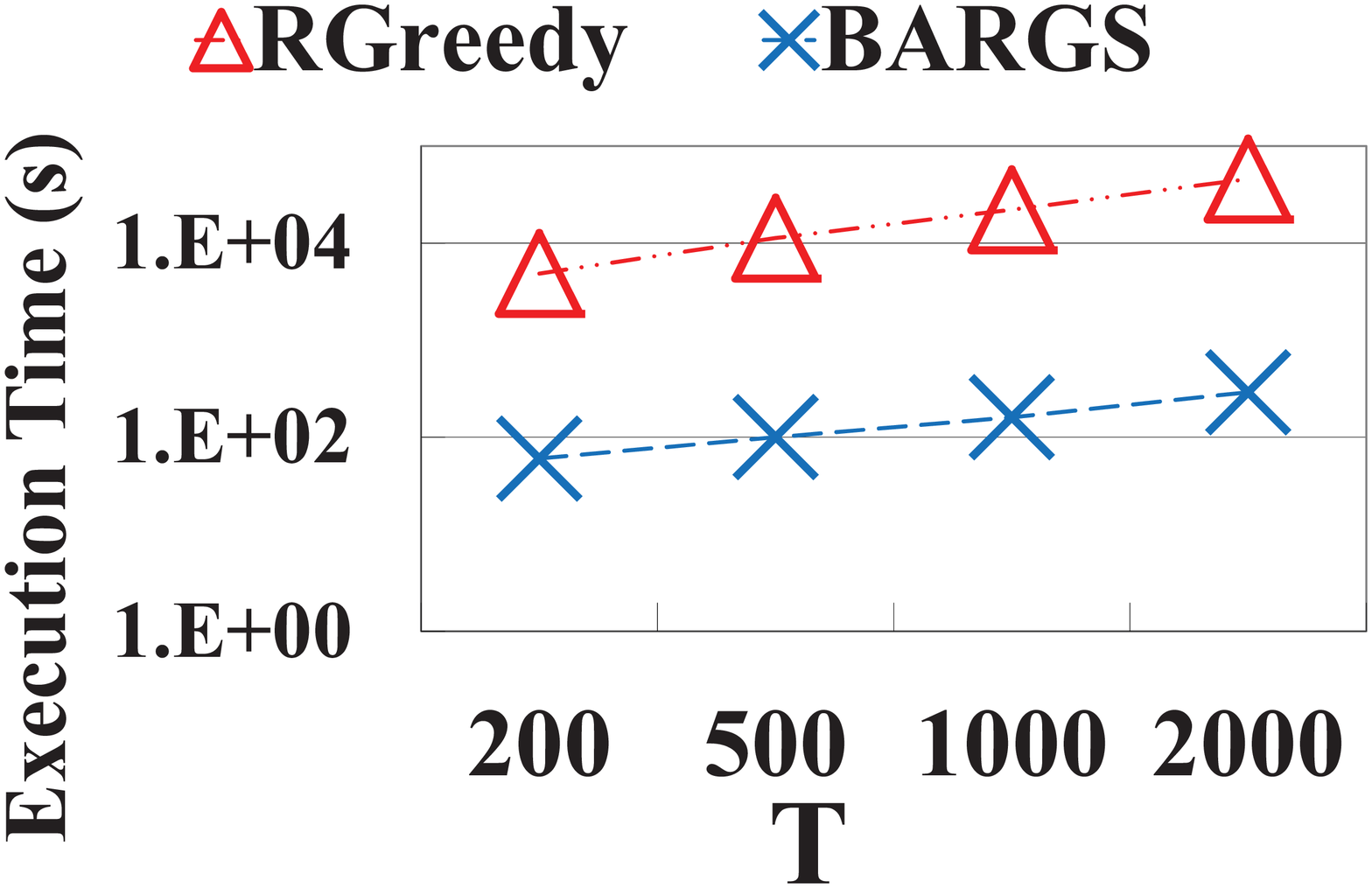} } 
\subfigure[] {\
\includegraphics[scale=0.14]{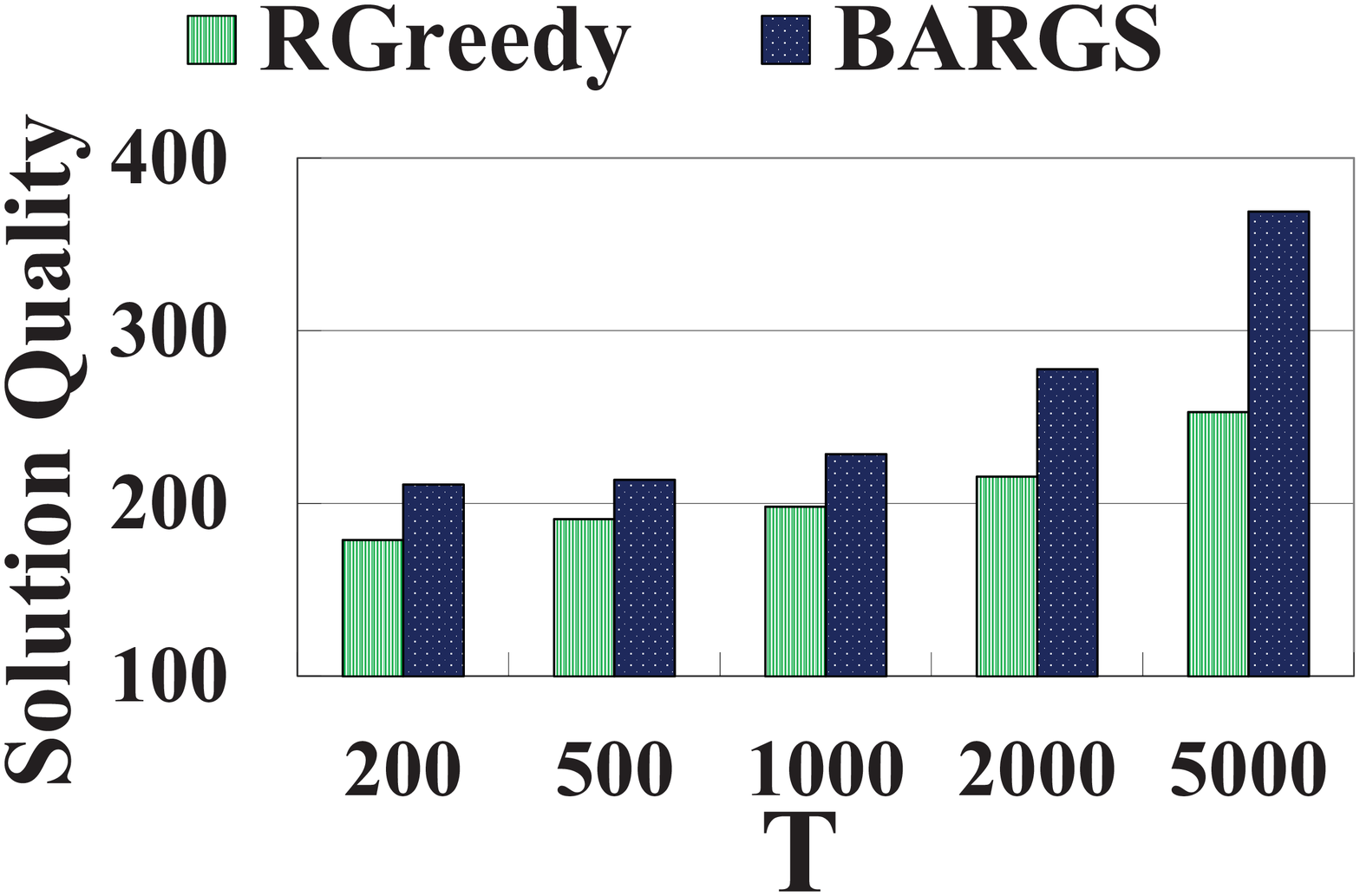} } 
\subfigure[] {\
\includegraphics[scale=0.14]{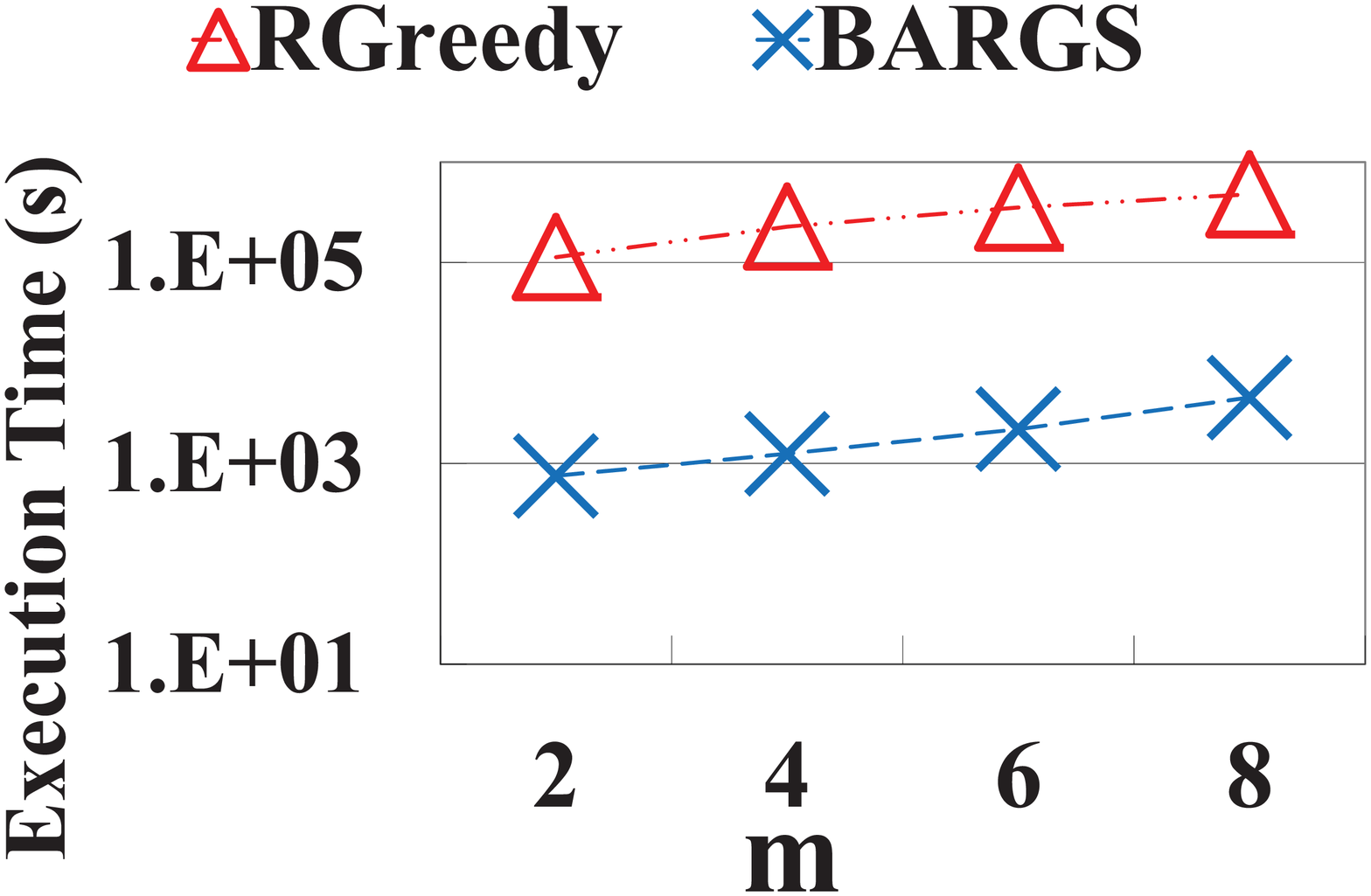} } 
\subfigure[] {\
\includegraphics[scale=0.14]{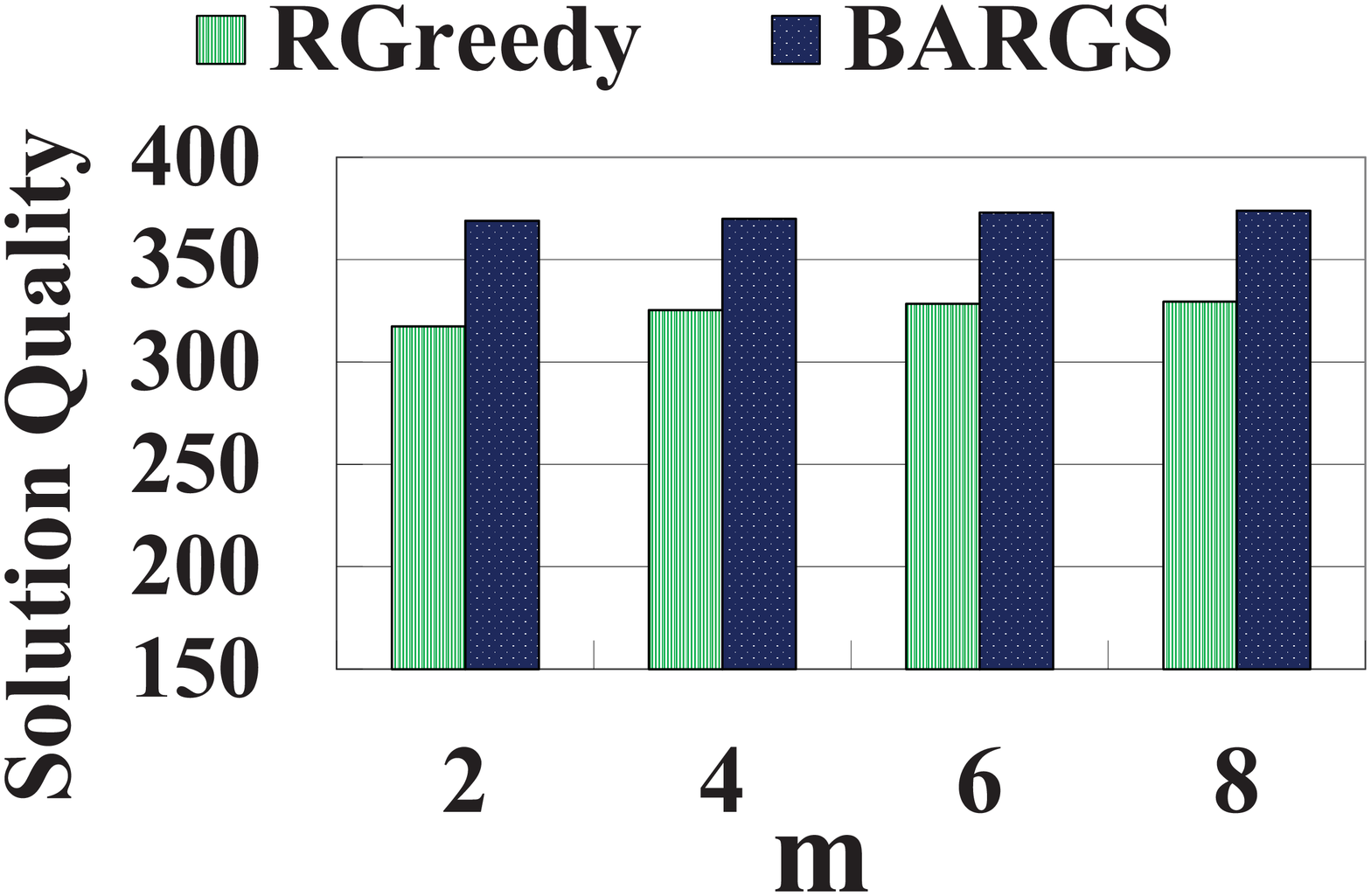} } 
\subfigure[] {\
\includegraphics[scale=0.14]{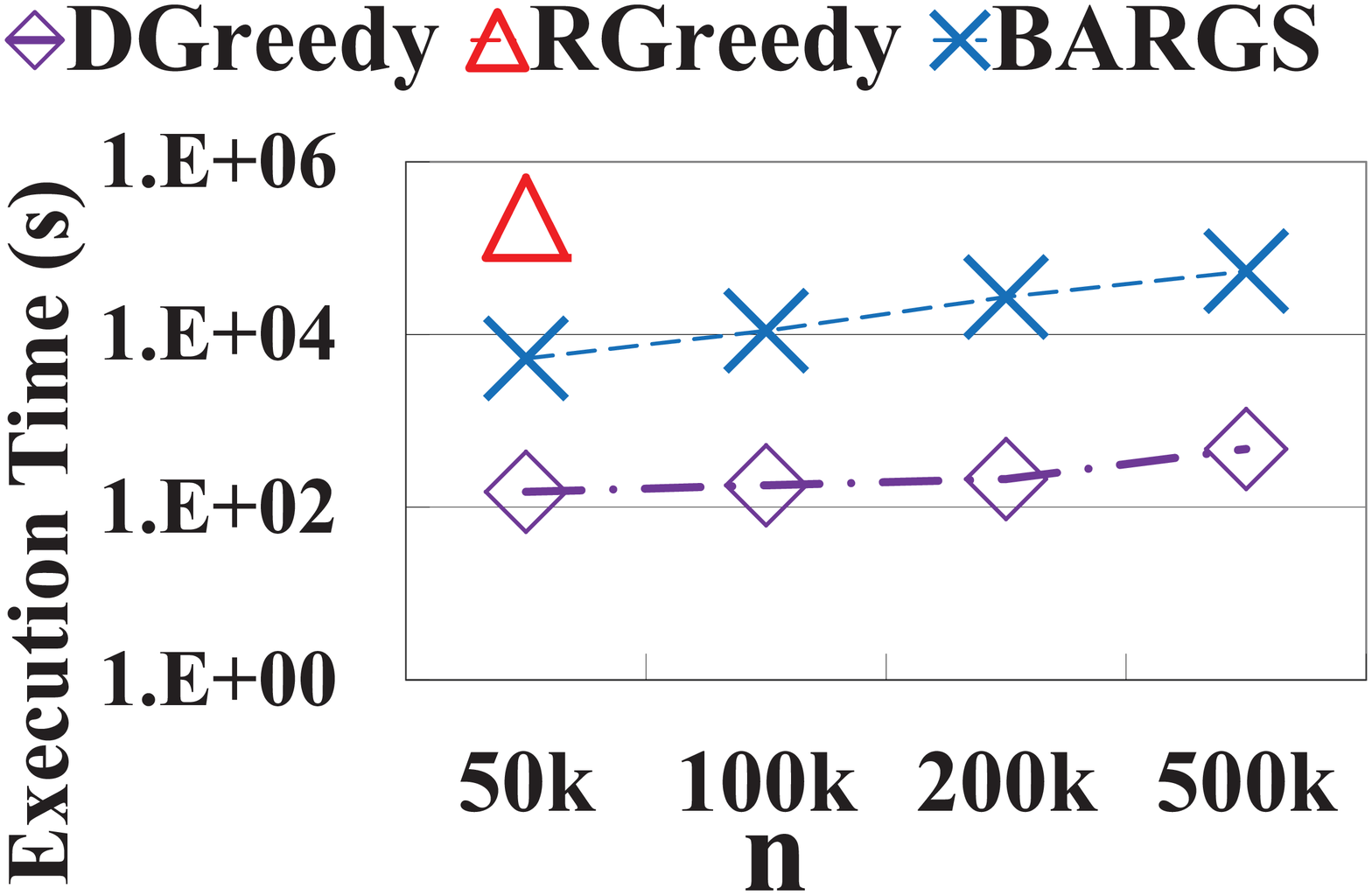} } 
\subfigure[] {\
\includegraphics[scale=0.14]{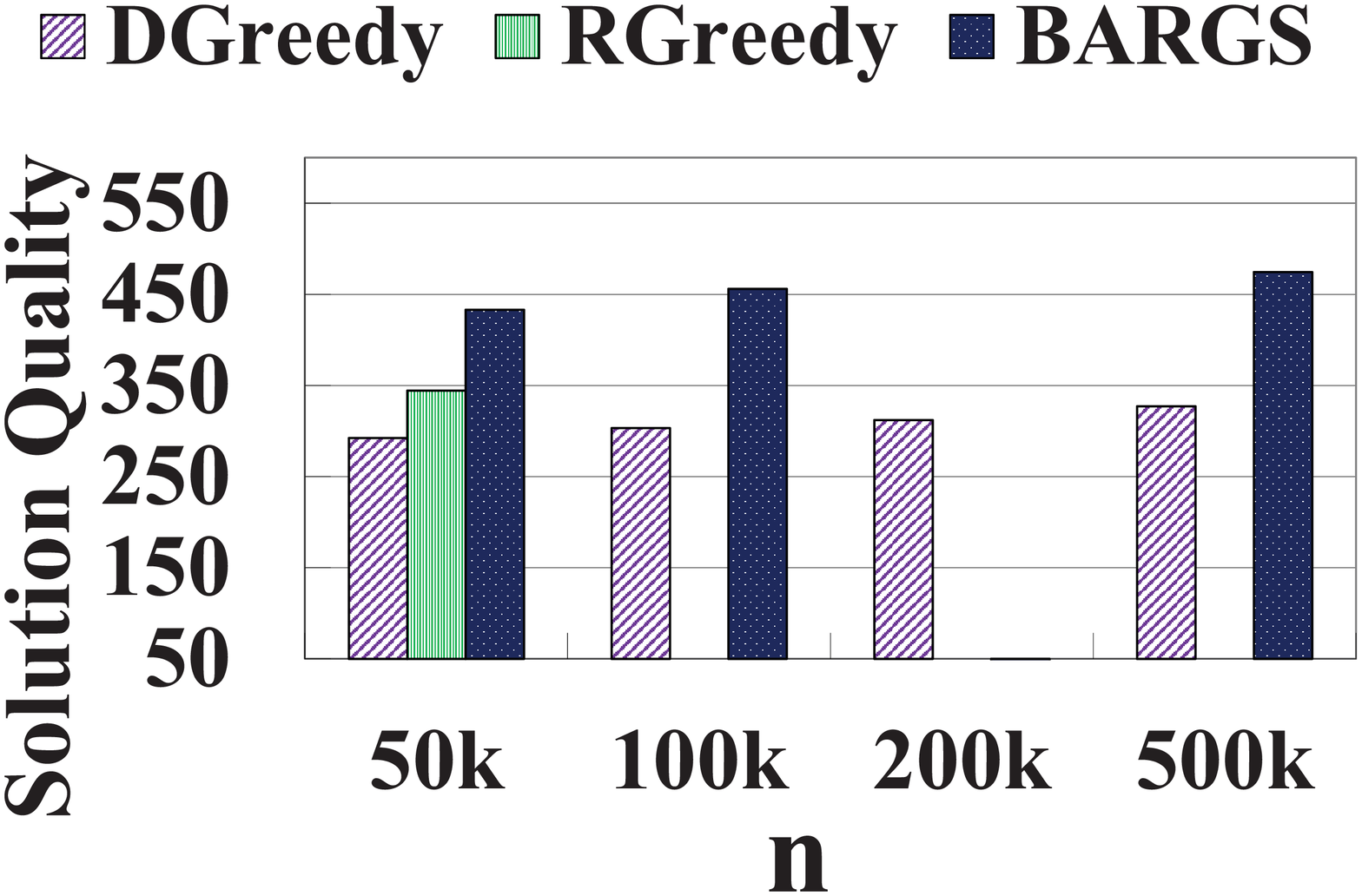} }
 \vspace{-12pt} 
\caption{Experimental results on Facebook and DBLP datasets}
\label{exp2}
\vspace{-18pt}
\end{figure}

Figure \ref{exp2}(a) compares the execution time of \emph{DGreedy}, \emph{%
RGreedy}, and \emph{BARGS} by sampling different numbers of nodes from Facebook
data. \emph{DGreedy} is always the fastest one since it is a deterministic
algorithm and generates only one final solution, whereas \emph{RGreedy} requires more than $10^{5}$
seconds. The results of \emph{RGreedy} do not return in 2 days as $n$
increases to $10000$. To evaluate the performance of \emph{BARGS} with
multi-threaded processing, Figure \ref{exp2}(b) shows that we can accelerate
the processing speed to $7.2$ times with $8$ threads. The
acceleration ratio is slightly lower than $8$ because OpenMP forbids
different threads to write at the same memory position at the same time.
Therefore, it is expected that \emph{BARGS} with parallelization is
promising to be deployed as a value-added \textit{cloud service}.

In addition to the running time, Figure \ref%
{exp2}(c) compares the solution quality of different approaches.
The results indicate that \emph{BARGS} outperforms \emph{DGreedy} and \emph{%
RGreedy}, especially under a large $n$. The group utility of \emph{BARGS} is
45\% better than the one from \emph{DGreedy} when $n=50000$. On the other
hand, \emph{RGreedy} outperforms \emph{DGreedy} since it has a chance to
jump out of the local optimal solution.

Figures \ref{exp2}(d) and (e) compare the execution time and solution
quality of two randomized approaches under different total computational
budgets, i.e., $T$. As $T$ increases, the solution quality of \emph{BARGS}
increases faster than that of \emph{RGreedy} because it can optimally allocate
the computation resources. Even though the solution quality of \emph{RGreedy}
is closer to \emph{BARGS} in some cases, \emph{BARGS} is much faster than \emph{%
RGreedy} by an order of $10^{-2}$.

Figures \ref{exp2}(f) and (g) present the execution time and solution
quality of \emph{RGreedy} and \emph{BARGS} with different numbers of start
nodes, i.e., $m$. The results show that
the solution quality in Figure \ref{exp2}(g) is almost the same as $m$
increases, demonstrating that it is sufficient for $m$ to be set as a
value smaller than $\frac{n}{k_{max}}$ as recommended by OCBA \cite{OCBA10}.
The running time of \emph{BARGS} for $m=2$ is only $60\%$ of the running time for $m=4$ as shown in Figure \ref{exp2}(f), while the solution quality remains almost the same.

\emph{BARGS} is also evaluated on the DBLP dataset. Figures \ref{exp2}(h)
and (i) show that \emph{BARGS} outperforms \emph{DGreedy} by $50\%$ and \emph{RGreedy} by $26\%
$ in solution quality when $n=500000$. \emph{BARGS} is still faster than \emph{RGreedy} by
an order of $10^{-2}$. However, \textit{RGreedy }runs faster on the DBLP
dataset than on the Facebook dataset, because the DBLP dataset is a sparser
graph with an average node degree of $3.66$. Therefore, the number of
candidate nodes to be chosen during the expansion of the partial solution in the DBLP dataset increases much more
slowly than in the Facebook dataset with an average node degree of $26.1$.
Nevertheless, \emph{RGreedy} is still not able to generate a solution for a
large network size $n$ due to its unacceptable efficiency.

\vspace{-3mm}
\section{Conclusion}
\label{Conclu}
\vspace{-1mm}
To the best of our knowledge, there is no real system or existing work in
the literature that addresses the issues of scale-adaptive group optimization for social activity planning based on topic interest, social tightness, and activity cost. To fill this research gap and
satisfy an important practical need, this paper formulated a new
optimization problem called PSGA to derive a set of attendees and maximize
the group utility. We proved that PSGA is NP-hard and devised a simple but
effective randomized algorithms, namely \emph{BARGS}, with a guaranteed performance bound. The user study demonstrated that the social groups
obtained through the proposed algorithm implemented in Facebook
significantly outperforms the manually configured solutions by users. This
research result thus holds much promise to be profitably adopted in social
networking websites as a value-added service.
\newpage
\appendix
\begin{algorithm}[H]
\caption{BARGS}
\label{BARGS-alg}
\begin{algorithmic}[1]
\renewcommand{\algorithmicrequire}{\textbf{Input:}}
\renewcommand{\algorithmicensure}{\textbf{Output:}}
\REQUIRE Graph $G(V,E)$, social network size $n$, activity cost function $C(k)$, maximum group size $k_{max}$, correctly select probability $P(CS)$, solution quality $Q$, percentile of CE $\rho$, and smoothing weighting $w$
\ENSURE The best group F generating maximum willingness
\STATE $c_{i}=\infty$, $d_{i}= 0$ for all $i$;
\STATE $m=\left\lceil \frac{n}{k_{max}}\right\rceil$, $w=0$;
\STATE Select $m$ candidate nodes to candidate set $\mathcal{M}$;
\STATE $T_{1}=\left\lceil m\frac{\ln (\frac{2(1-P(CS))}{m-1})}{\ln \alpha }\right\rceil$;
\STATE Find the number of stages $r$ by first consulting $N_{b}$ table with solution $q$, and $r$=$\left\lceil \frac{4N_{b}}{T_{1}}-\frac{4k}{n}+1\right\rceil$;
\FOR {$t=1$ to $r$}
     \IF {$t=1$}
       \FOR {$i=1$ to $m$}
         \STATE $A_{i}=\frac{T_{1}}{m}$;
         \STATE Set the node selection probability vector $p_{i,t}$ as uniform;
       \ENDFOR
    \ELSE
           \STATE $A_{total}=0$;
           \FOR {$i=1$ to $m$}
             \STATE $A_{i}$= $\frac{1}{2}(\frac{d_{i}-c_{b}}{d_{b}-c_{b}})^{N_{b}}$;
             \STATE $A_{total}$=$A_{total}$+$A_{i}$;
           \ENDFOR
           \STATE  $A_{i}$= $T_{1}A_{i}$/$A_{total}$; 
    \ENDIF

    \FOR {$i=1$ to $m$}
        \STATE $V_{S}=\mathcal{M}_{i}$
        \STATE $V_{A}=\emptyset$
        \STATE $X=\emptyset$
        \FOR{$x=1$ to $A_{i}$}
           \STATE $V_{A}=N(\mathcal{M}_{i})$
           \FOR{$k=1$ to $k_{max}-1$}
              \STATE Random select a node $v$ in $V_{A}$ in accordance with $p_{i,k,t}$ to $V_{S}$;
              \STATE $V_{A}=V_{A} \cup N(v)$
           \ENDFOR
             \STATE $u=U(V_{S})$;
             \STATE $X.add(V_{S},u)$;
              \IF {$u>d_{i,k}$}
                \STATE $d_{i,k}=u$;
              \ENDIF
              \IF {$w<c_{i,k}$}
                \STATE $c_{i,k}=w$;
              \ENDIF
              \IF {$w>W(F)$}
                \STATE $b=i$;
                \STATE $F=V_{S}$;
              \ENDIF
        \ENDFOR
    \ENDFOR
                \COMMENT{Update node selection probability $p_{i,k,t+1}$}
                \STATE X=$DescendingSort(X,u)$;
                \IF {$\gamma_{t}>X_{(\left\lceil \rho A_{i}\right\rceil)}.w$}
                    \STATE $\gamma_{t+1}=\gamma_{t}$; 
                \ELSE
                    \STATE $\gamma_{t+1}=X_{(\left\lceil \rho A_{i}\right\rceil)}.w$; 
                \ENDIF             
                \FORALL {Sample $x$ in $X$}
                      \IF {$x.u>\gamma_{t+1}$}
                        \FORALL {$v_{j} \in x$}
                           \STATE  $p_{i,k,t+1,j}=p_{i,k,t+1,j}+1$;
                        \ENDFOR
                      \ENDIF
                \ENDFOR
                \FOR{$j=1$ to $n$}
                   \STATE  $p_{i,,k,t+1,j}=p_{i,k,t+1,j}/\left\lceil \rho A_{i}\right\rceil$;
                   \STATE  $p_{i,k,j,t+1}=w p_{i,t+1,j}+(1-w)p_{i,k,t,j}$;
                \ENDFOR
\ENDFOR

    \STATE Output $F$;
\end{algorithmic}
\end{algorithm}

\bibliographystyle{abbrv}
\bibliography{hhshuai_2015}

\end{document}